\documentclass{llncs}

\usepackage{amssymb}
\usepackage[english]{babel}
\usepackage[utf8x]{inputenc}
\usepackage{amsmath}
\usepackage{graphicx}
\usepackage{pxfonts}
\usepackage[colorinlistoftodos]{todonotes}
\usepackage{float}
\usepackage{epstopdf}
\usepackage{subcaption}
\usepackage{llncsdoc}

\usepackage{times}
\usepackage{booktabs}
\usepackage{rotating}
\usepackage{xcolor}
\definecolor{darkred}{rgb}{0.7, 0.0, 0.0}
\usepackage{caption}
\usepackage{listings}
\usepackage{xspace}
\usepackage{hyperref}
\usepackage{amsfonts}
\usepackage{amsmath}
\newcommand{\comment}[1]{}

\usepackage{amssymb}
\usepackage{temporal}
\usepackage{hyperref}
\usepackage{breakurl}
\usepackage{stmaryrd}
\usepackage{enumerate}
\usepackage{wrapfig} 

\newtheorem{proof}{Proof}

\newcounter{exacounter}
\newenvironment{exa}{
\refstepcounter{exacounter}
\smallskip\noindent
\textbf{Example \theexacounter.}
}{\vspace{2mm}}

\usepackage{tikz}
\usetikzlibrary{arrows,automata}
\tikzset{initial text={}}
\tikzset{every picture/.style=semithick} 
\tikzset{>=stealth'} 
\tikzset{->} 
\tikzset{shorten >=1pt}

\newcommand{\buchi}{B\"uchi\xspace}
\newcommand{\win}{\mathsf{win}}
\newcommand{\B}{\mathbb{B}}
\newcommand{\N}{\mathbb{N}}

\newcommand{\design}{\mathcal{D}}
\newcommand{\shield}{\mathcal{S}}
\newcommand{\game}{\mathcal{G}}

\newcommand{\gstates}{G}
\newcommand{\ginit}{g_0}
\newcommand{\states}{Q}
\newcommand{\init}{q_0}
\newcommand{\din}{I}
\newcommand{\dinalph}{\Sigma_I}
\newcommand{\dinletter}{{\sigma_I}}
\newcommand{\dintrace}{{\overline{\sigma_I}}}
\newcommand{\dout}{O}

\newcommand{\doutalph}{\Sigma_O}
\newcommand{\doutletter}{{\sigma_O}}
\newcommand{\douttrace}{{\overline{\sigma_O}}}
\newcommand{\dalph}{\Sigma}
\newcommand{\dletter}{\sigma}
\newcommand{\dtrace}{\overline{\dletter}}
\newcommand{\lang}{L}
\newcommand{\langset}{\mathcal{L}}
\newcommand{\spec}{\varphi}

\newcommand{\kin}{\!\in\!}
\newcommand{\err}{{\color{darkred}$\lightning$}}
\newcommand{\comp}{\circ}

\begin{document}

\title{Synthesis of Admissible Shields%
%
\thanks{This work was supported in part by the Austrian Science Fund
  (FWF) through the research network RiSE (S11406-N23),
  and by the European Commission through the project IMMORTAL (644905).}}

\author{Laura Humphrey\inst{1}, Bettina K\"onighofer\inst{2},
        Robert K\"onighofer\inst{2}, Ufuk Topcu\inst{3}}

\institute{
   $^1$ Control Science Center of Excellence, AFRL\\
   $^2$ IAIK, Graz University of Technology,  Austria\\
   $^3$ University of Texas at Austin, USA\\
          }
\maketitle

\begin{abstract}
\emph{Shield synthesis} is an approach to enforce a set of
safety-critical properties of a reactive system at runtime. A shield
monitors the system and corrects any erroneous output values
instantaneously. The shield deviates from the given outputs as little
as it can and recovers to hand back control to the system as soon as
possible.
This paper takes its inspiration from a case study on mission
planning for unmanned aerial vehicles (UAVs) in which
\emph{$k$-stabilizing} shields, which guarantee recovery in a finite
time, could not be constructed.  We introduce the notion of
\emph{admissible} shields, which improves \emph{$k$-stabilizing}
shields in two ways: (1) whereas $k$-stabilizing shields take an
adversarial view on the system, admissible shields take a
collaborative view. That is, if there is no shield that guarantees
recovery within $k$ steps regardless of system behavior, the
admissible shield will attempt to work with the system to recover as
soon as possible. (2) Admissible shields can handle system failures
during the recovery phase.
In our experimental results we show that for UAVs, we can generate
admissible shields, even when $k$-stabilizing shields do not exist.
\end{abstract}

\section{Introduction}

Technological advances enable the development of increasingly sophisticated systems.
Smaller and faster microprocessors, wireless networking, 
and
new theoretical results in areas such as machine learning and intelligent control
are paving the way for transformative technologies across a variety of domains
-- self-driving cars that have the potential to reduce accidents, traffic, energy consumption, and pollution;
and unmanned systems that can safely and efficiently operate on land, under water, in the air, and in space.
However, in each of these domains, concerns about
safety are being raised \cite{4772749},\cite{dalamagkidis2011integrating}.
Specifically, there is a concern that due to the complexity of such systems,
traditional test and evaluation approaches will not be sufficient for finding errors,
and alternative approaches such as those provided by formal methods are needed \cite{Lygeros96}.

Formal methods are often used to verify systems at design time, but this is not always realistic.
Some systems are simply too large to be fully verified.
Others, especially systems that operate in rich dynamic environments or those that continuously adapt their behavior through methods such as machine learning, cannot be fully modeled at design time.
Still others may incorporate components that have not been previously verified and cannot be modeled,
e.g., proprietary components or pre-compiled code libraries.

Also, even systems that have been fully verified at design time may be subject to external faults such as
those introduced by unexpected hardware failures or human inputs.
One way to address this issue is to model nondeterministic behaviours
(such as faults) as disturbances,
and verify the system with respect to this disturbance model~\cite{Mancini14}.
However, it is impossible to model all potential unexpected behavior at design time.

An alternative in such cases is to perform \emph{runtime verification} to detect violations
of a set of specified properties while a system is executing \cite{leucker2009}.
An extension of this idea is to perform \emph{runtime enforcement} of specified properties,
in which violations are not only detected but also overwritten
in a way that specified properties are maintained.

A general approach for runtime enforcement of specified properties is
\emph{shield synthesis}, in which a shield monitors the system and
instantaneously overwrites incorrect outputs.  A shield must ensure
both \emph{correctness}, i.e., it corrects system outputs such that
all properties are always satisfied, as well as \emph{minimum
  deviation}, i.e., it deviates from system outputs only if necessary
and as rarely as possible.  The latter requirement is important
because the system may satisfy additional noncritical properties that
are not considered by the shield but should be retained as much as
possible.

Bloem et al.~\cite{BloemKKW15} proposed the notion of $k$-stabilizing
shields. Since we are given a safety specification, we can identify
wrong outputs, that is, outputs after which the specification is
violated (more precisely: after which the environment can force the
specification to be violated). A wrong trace is then a trace that ends
in a wrong output. The idea of shields is that they may modify the
outputs so that the specification always holds, but that such
deviations last for at most $k$ consecutive steps after a wrong
output.  If a second violation happens during the $k$-step recovery
phase, the shield enters a mode where it only enforces correctness,
but no longer minimizes the deviation.  This proposed approach has two
limitations with significant impact in practice.  (1) The
$k$-stabilizing shield synthesis problem is unrealizable for many
safety-critical systems, because a finite number of deviations cannot
be guaranteed. (2) $k$-stabilizing shields make the assumption that
there are no further system errors during the recovery phase.

In this paper, we introduce \emph{admissible} shields, which overcome
the two issues of $k$-stabilizing shields. To address shortcoming (1),
we guarantee the following: (a) Admissible shields are subgame
optimal. That is, for any wrong trace, if there is a finite number $k$
of steps within which the recovery phase can be guaranteed to end, the
shield will always achieve this. (b) The shield is \emph{admissible},
that is, if there is no such number $k$, it always picks a deviation
that is optimal in that it ends the recovery phase as soon as possible for
some possible future inputs. (This is defined in more detail
below.) As a result, admissible shields work well in settings in which
finite recovery can not be guaranteed, because they guarantee
correctness and may well end the recovery period if the system does
not pick adversarial outputs. To address shortcoming (2), admissible
shields allow arbitrary failure frequencies and in particular failures
that arrive during recovery, without losing the ability to recover.

As a second contribution, we demonstrate the use of admissible shields
through a case study involving mission planning for an unmanned aerial vehicle (UAV).
Manually creating and executing mission plans that meet mission
objectives while addressing all possible contingencies is
a complex and error-prone task. Therefore, having a shield that
changes the mission only if absolutely necessary to enforce certain
safety properties has the potential to lower the burden on human
operators, and ensures safety during mission execution. We show that
admissible shields are applicable in this setting, whereas
$k$-stabilizing shields are not.

\textbf{Related Work:}
Our work builds on synthesis of reactive systems \cite{Pnueli1989}, \cite{BloemJPPS12}
and reactive mission plans~\cite{EhlersKB15} from formal specifications,
and our method is related to synthesis of robust~\cite{BloemCGHHJKK14} and
error-resilient~\cite{EhlersT14} systems.
However, our approach differs in that we do not synthesize an entire system,
but rather a shield that considers only a small set of properties and corrects
the output of the system at runtime.
Li et al.~\cite{LiSSS14} focused on the problem of synthesizing a semi-autonomous controller
that expects occasional human intervention for correct operation.
A human-in-the-loop controller monitors past and current information about the system and its
environment. The controller invokes the human operator only when it is necessary,
but as soon as a specification is violated ahead of time, such that the human operator
has sufficient time to respond.
Similarly, our shields monitor the behavior of systems at run time,
and interfere as little as possible.
Our work relates to more general work on runtime enforcement of properties \cite{FalconeFM12}, but
shield synthesis~\cite{BloemKKW15} is the first appropriative work for reactive systems,
since shields act on erroneous system outputs immediately without delay.
While \cite{BloemKKW15} focuses on shield synthesis for
systems assumed to make no more than one error every $k$ steps, this work
assumes only that systems generally have cooperative behavior with respect to the shield,
i.e., the shield ensures a finite number of deviations if the system chooses certain outputs.
This is similar in concept to cooperative synthesis as considered in \cite{BloemEK15},
in which a synthesized system has to satisfy a set of properties (called guarantees)
only if certain environment assumptions hold.
The authors present a synthesis procedure that
maximizes the cooperation between system and environment for satisfying
both guarantees and assumptions as far as possible.

\textbf{Outline:} In what follows, we begin in Section \ref{sec:ex} by motivating the need for
admissible shields through a case study involving
mission planning for a UAV.
In Sections \ref{sec:prelim}, \ref{sec:def}, \ref{sec:sol}, we define preliminary concepts,
review the general shield synthesis framework, and describe our approach for synthesizing admissible shields.
Section \ref{sec:exp} provides experimental results,
and Section \ref{sec:conc} concludes.

\section{Motivating Example}
\label{sec:ex}

In this section, we apply shields on a scenario in which a UAV must
maintain certain properties while performing a surveillance mission in a dynamic environment.
We show how a shield can be used to enforce the desired properties,
where a human operator in conjunction with a lower-level autonomous planner
is considered as the reactive system that sends commands to the UAV's
autopilot.  We discuss how we would intuitively want a shield to
behave in such a situation. We show that the \emph{admissible} shields
provide the desired behaviors and address the limitations of $k$-stabilizing shields.

To begin, note that a common UAV control architecture consists of a ground control station
that communicates with an autopilot onboard the UAV \cite{chao_2010}.
The ground control station receives and displays updates from the autopilot on the UAV's state,
including position, heading, airspeed, battery level, and sensor imagery.
It can also send commands to the UAV's autopilot,
such as waypoints to fly to.
A human operator can then use the ground control station to plan waypoint-based
routes for the UAV, possibly making modifications during mission execution to respond to events observed through the UAV's sensors.
However, mission planning and execution can be very workload intensive,
especially when operators are expected to control multiple UAVs simultaneously \cite{donmez_2010a}.
To address this issue, methods for UAV command and control have been explored
in which operators issue high-level commands, and automation carries out
low-level execution details.

Several errors can occur in this type of human-automation paradigm \cite{chen_2012}.
For instance, in issuing high-level commands to the low-level planner,
a human operator might neglect required safety properties due to
high workload, fatigue, or an incomplete understanding
of exactly how the autonomous planner might execute the command.
The planner might also neglect these safety properties either because of software errors or by design. Waypoint commands issued by the operator or planner
could also be corrupted by software that translates
waypoint messages between ground station and autopilot specific formats or
during transmission over the communication link.

As the mission unfolds, waypoint commands will be sent periodically to the autopilot.
If a waypoint that violates the properties is received,
a shield that monitors the system inputs and can overwrite the waypoint
outputs to the autopilot would be able to make
corrections to ensure the satisfaction of the desired properties.

Consider the mission map in Fig.~\ref{fig:map} \cite{Feng16}, which contains three tall buildings (illustrated as blue blocks),
over which a UAV should not attempt to fly.
It also includes two unattended ground sensors (UGS) that provide data
on possible nearby targets, one at location $loc_1$ and one at $loc_x$,
as well as two locations of interest, $loc_y$ and $loc_z$.
The UAV can monitor $loc_x$, $loc_y$, and $loc_z$ from several nearby vantage points.
The map also contains a restricted operating zone
(ROZ), illustrated with a red box, in which flight might be dangerous,
and the path of a possible adversary that should be avoided
(the pink dashed line).
Inside the communication relay region (large green area), communication links
are highly reliable. Outside this region,
communication relies on relay points with lower reliability.
\begin{figure}[tb]
\centering
\includegraphics[width=3in]{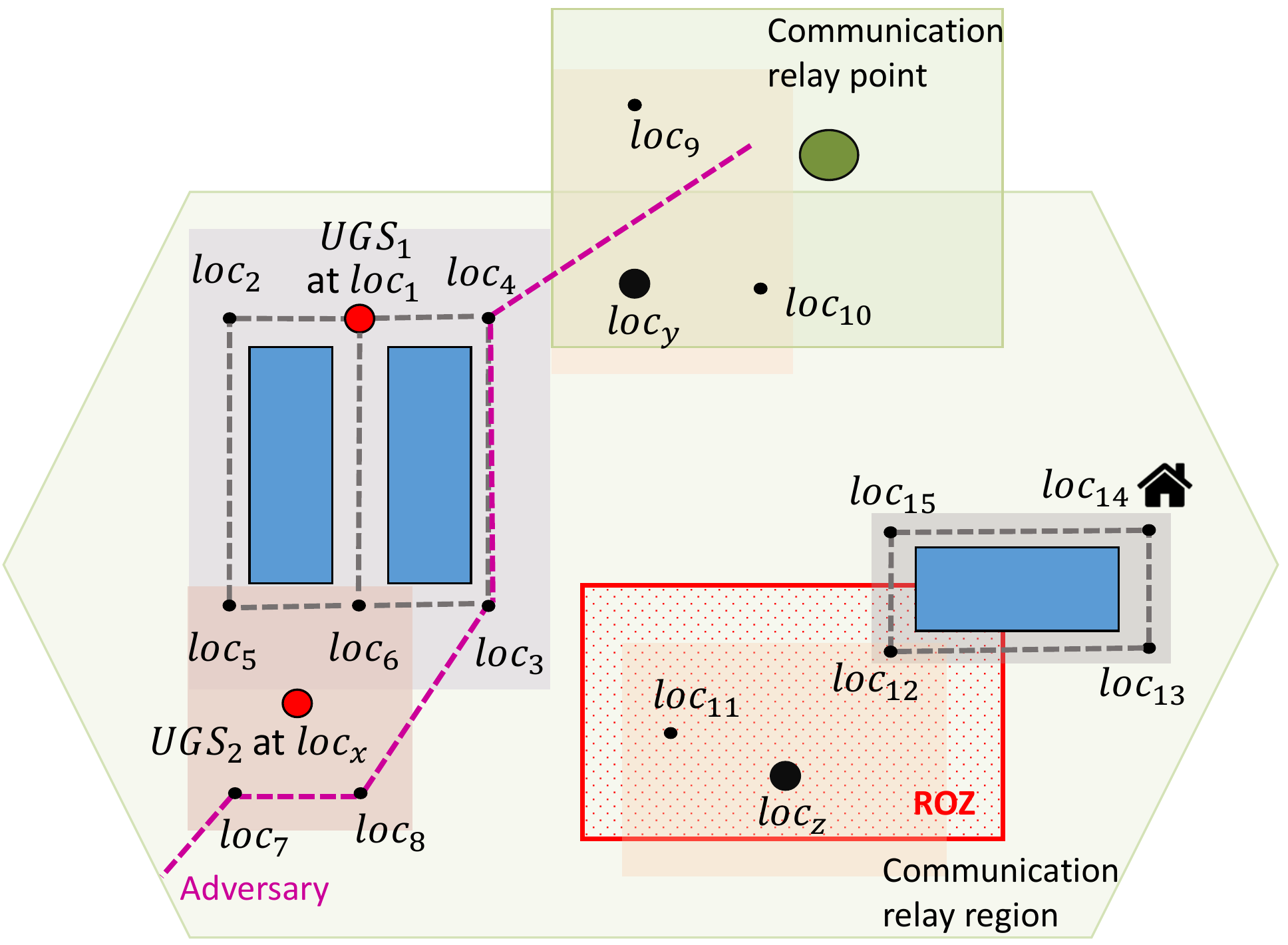}
\caption{A map for UAV mission planning.}
\label{fig:map}
\end{figure}
%
%
Given this scenario, properties of interest include: 

\begin{enumerate}
  \item \textbf{Connected waypoints.} \label{connected}
  The UAV is only allowed to fly to directly connected waypoints.
  \item \textbf{No communication.} The UAV is not allowed to stay in
  a location with reduced communication reliability.
  \item \textbf{Restricted operating zones.} \label{ROZ}
  The UAV has to leave a ROZ within 2 time steps.
  \item \textbf{Detected by an adversary.} Locations on the
  adversary's path cannot be visited more than once over any window of 3 time steps.\label{adversary}
  \item \textbf{UGS.} If a UGS reports a possible nearby target, the UAV
  should visit a respective waypoint within 7 steps
  (for $UGS_1$ visit $loc_1$, for $UGS_2$ visit $loc_5$, $loc_6$, $loc_7$, or $loc_8$).\label{ugs}
  \item \textbf{Go home.} Once the UAV's battery is low, it
  should return to a designated landing site at $loc_{14}$ within 10 time steps.\label{home}
\end{enumerate}
The task of the shield is to ensure these properties during operation.
In this setting, the operator in conjunction with a lower-level planner
acts as a reactive system that responds to mission-relevant inputs;
in this case data from the UGSs and a signal indicating whether the battery is low.
In each step, the next waypoint is sent to the autopilot, which is
encoded in a bit representation via outputs
$l_4$, $l_3$, $l_2$, and $l_1$.
We attach the shield as shown in Fig.~\ref{fig:attach_shield}.
The shield monitors mission inputs and waypoint outputs,
correcting outputs immediately if a violation of the safety properties becomes unavoidable.

We represent each of the properties by a safety automaton, the
product of which serves as the shield specification.
Fig.~\ref{fig:model_map} models the ``connected waypoints'' property,
where each state represents a waypoint with the same number.
Edges are labeled by the values of the variables $l_4\dots l_1$.
For example, the edge leading from state $s_5$
to state $s_6$ is labeled by
$\neg l_4 l_3 l_2 \neg l_1$.
For clarity, we drop the labels of edges in Fig.~\ref{fig:model_map}.
The automaton also includes an error state, which is not shown.
Missing edges lead to this error state, denoting forbidden situations.

How should a shield behave in this scenario?
If the human operator wants to monitor
a location in a ROZ, he or she would like to simply command the UAV
to ``monitor the location in the ROZ and stay there'',
with the planner handling the execution details.
If the planner cannot do this while
meeting all the safety properties, it is appropriate for the shield to revise its outputs.
Yet, the operator would still expect his or her commands to be followed
to the maximum extent possible, leaving the ROZ when necessary and returning whenever possible. Thus,
the shield should minimize deviations from the operator's
directives as executed by the planner.

\begin{figure}[tb]
\vspace{-18pt}
\begin{minipage}{\linewidth}
  \centering
  \begin{minipage}{0.42\linewidth}
    \begin{figure}[H]
      \centering
\includegraphics[width=2in]{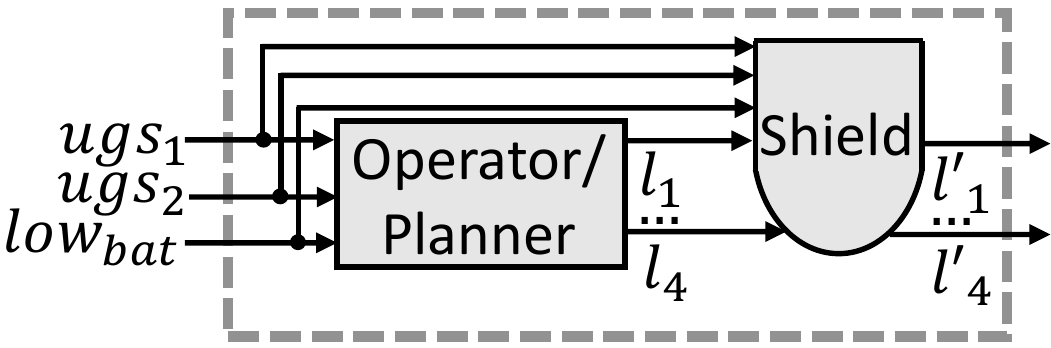}
\caption{The interaction between the operator/planner (acting as a reactive system) and the shield.}
\label{fig:attach_shield}
    \end{figure}
  \end{minipage}
  \hspace{0.02\linewidth}
  \begin{minipage}{0.5\linewidth}
    \begin{figure}[H]
      \centering
\includegraphics[width=2.3in]{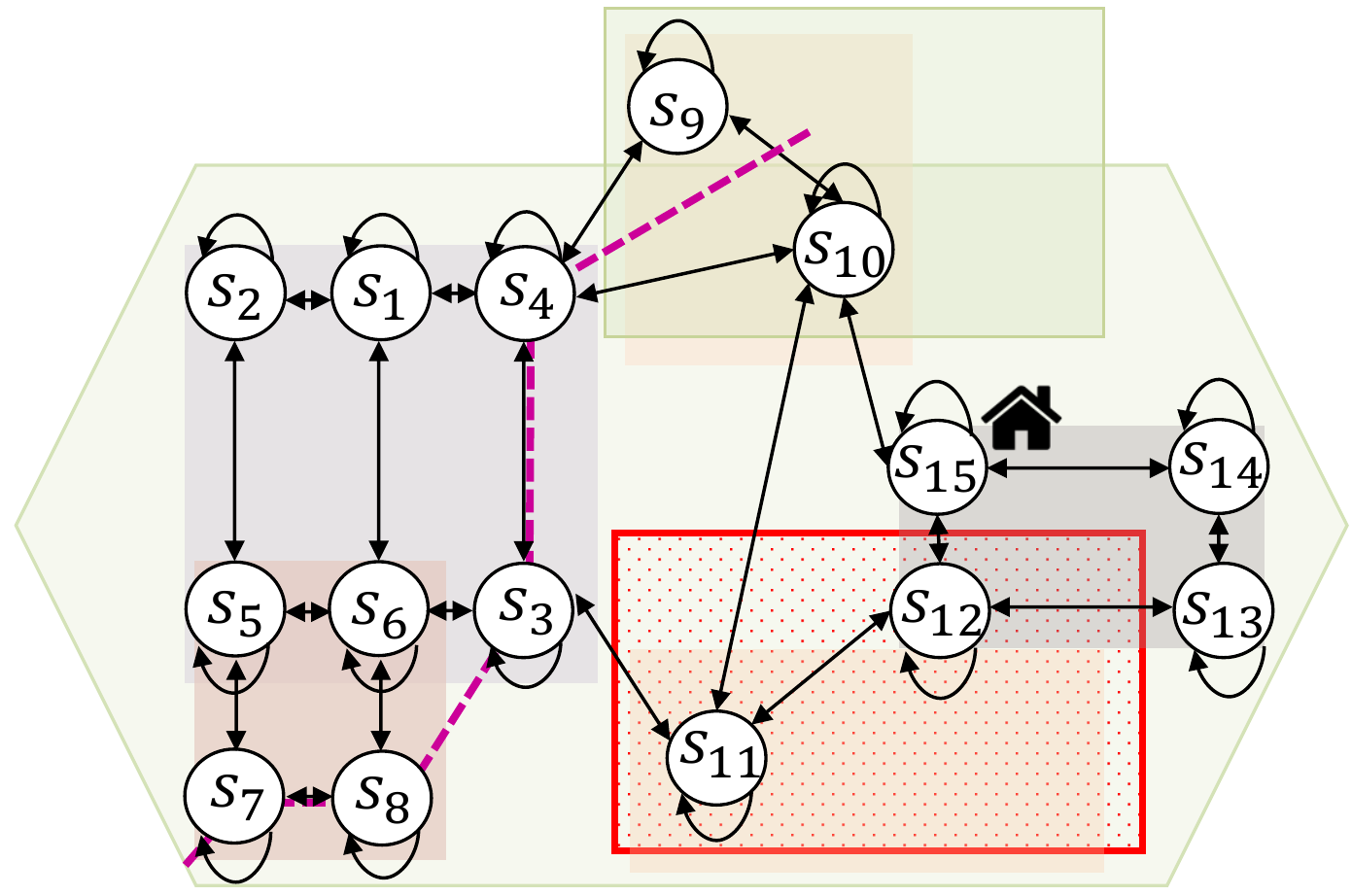}
\caption{Safety automaton of Property~\ref{connected} over the map in Fig.~\ref{fig:map}.}
\label{fig:model_map}
    \end{figure}
  \end{minipage}
\end{minipage}
\end{figure}

\textbf{Using a $k$-stabilizing shield.}
As a concrete example, assume the UAV is currently at $loc_3$,
and the operator commands it to monitor $loc_{12}$.
The planner then sends commands to fly to $loc_{11}$ then $loc_{12}$,
which are accepted by the shield.
The planner then sends a command to loiter at $loc_{12}$,
but the shield must overwrite it to maintain Property~\ref{ROZ},
which requires the UAV to leave the ROZ within two time steps.
The shield instead commands the UAV to go to $loc_{15}$.
Suppose the operator then commands the UAV to fly to $loc_{13}$, while
the planner is still issuing commands as if the UAV is at $loc_{12}$.
The planner then commands the UAV to fly to $loc_{13}$, but
since the actual UAV cannot fly from $loc_{15}$ to $loc_{13}$ directly, the
shield directs the UAV to $loc_{14}$ on its way to $loc_{15}$.
The operator might then respond to a change in the mission and
command the UAV fly back to $loc_{12}$, and the shield again deviates
from the route assumed by the planner, and directs the UAV back to
$loc_{15}$, and so on.
Therefore, a single specification violation can lead to an infinitely long deviation between
the UAV's actual position and the UAV's assumed position.
A $k$-stabilizing shield is allowed to deviate from the planner's commands
for at most $k$ consecutive time steps.
Hence, no $k$-stabilizing shield exists.

\textbf{Using an admissible shield.}
Recall the situation in which the shield
caused the actual position of the UAV to ``fall behind'' the position assumed by the planner,
so that the next waypoint the planner issues is two or more steps away from the UAV's current
waypoint position.
The shield should then implement a best-effort strategy to ``synchronize'' the UAV's
actual position with that assumed by the planner.
Though this cannot be guaranteed,
the operator and planner are not adversarial towards the shield,
so it will likely be possible to achieve this re-synchronization,
for instance when the UAV goes back to a previous waypoint or remains at the current waypoint
for several steps.
This possibility motivates the concept of an \emph{admissible} shield.
Assume that the actual position of the UAV is $loc_{14}$ and the its assumed position
is $loc_{13}$. If the operator commands the
UAV to loiter at $loc_{13}$, the shield will be able to catch up
with the state assumed by the planner and to end the deviation by the next
specification violation.

\section{Preliminaries}
\label{sec:prelim}

We denote the Boolean domain by $\B=\{\true,\false\}$, the set of
natural numbers by $\N$, and abbreviate $\N\cup\{\infty\}$ by
$\N^\infty$.
We consider a reactive system with a finite set
$\din=\{i_1,\ldots,i_m\}$ of Boolean inputs and a finite set
$\dout=\{o_1,\ldots,o_n\}$ of Boolean outputs.  The input alphabet is
$\dinalph=2^\din$, the output alphabet is $\doutalph=2^O$, and
$\dalph=\dinalph \times \doutalph$. The set of finite (infinite) words
over $\dalph$ is denoted by $\dalph^*$ ($\dalph^\omega$), and
$\dalph^{\infty} = \dalph^* \cup \dalph^\omega$.  We will also refer
to words as \emph{(execution) traces}.  We write $|\dtrace|$ for the
length of a trace $\dtrace\in \dalph^*$. For $\dintrace = x_0
x_1 \ldots \in \dinalph^\infty$ and $\douttrace = y_0 y_1 \ldots \in
\doutalph^\infty$, we write $\dintrace || \douttrace$ for the
composition $(x_0,y_0) (x_1,y_1) \ldots \in \dalph^\infty$. A set $\lang
\subseteq  \dalph^\infty$ of words is called a \emph{language}.
We denote the set of all languages as $\langset = 2^{\dalph^\infty}$.

\noindent
\textbf{Reactive Systems.}
A \emph{Mealy machine} (reactive system, design) is a 6-tuple $\design  = (\states, \init, \\\dinalph, \doutalph, \delta, \lambda)$, where $\states$ is a
finite set of states, $\init\in \states$ is the initial state, $\delta:
\states \times \dinalph \rightarrow \states$ is a complete transition
function, and $\lambda: \states \times \dinalph \rightarrow \doutalph$
is a complete output function.  Given the input trace $\dintrace = x_0
x_1 \ldots \in \dinalph^\infty$, the system $\design$ produces the
output trace $\douttrace = \design(\dintrace) = \lambda(q_0, x_0)
\lambda(q_1, x_1) \ldots \in \doutalph^\infty$, where $q_{i+1} =
\delta(q_i, x_i)$ for all $i \ge 0$.  The set of words produced by
$\design$ is denoted $\lang(\design) = \{\dintrace || \douttrace \in
\dalph^\infty \mid \design(\dintrace) = \douttrace\}$. 

Let $\design = (\states, \init, \dinalph, \doutalph, \delta, \lambda)$
and $\design' = (\states', \init', \dalph, \doutalph, \delta',
\lambda')$ be reactive systems.
A serial composition of $\design$ and $\design'$ is realized if the input and output
of $\design$ are fed to $\design'$. We denote such composition as $\design \comp \design'=(\hat{\states}, \hat{\init}, \dinalph, \doutalph, \hat{\delta},
\hat{\lambda})$, where
$\hat{\states} = \states \times \states'$,
$\hat{\init} = (\init, \init')$,
$\hat{\delta}((q,q'),\dinletter) = (\delta(q,\dinletter),
                   \delta'(q',(\dinletter,\lambda(q,\dinletter))))$, and
$\hat{\lambda}((q,q'),\dinletter) =
                     \lambda'(q',(\dinletter,\lambda(q,\dinletter)))$.

\noindent
\textbf{Specifications.}
A \emph{specification} $\spec$ is a set $\lang(\spec) \subseteq
\dalph^\infty$ of allowed traces.
$\design$ \emph{realizes} $\spec$, denoted by $\design \models \spec$, iff
$\lang(\design) \subseteq \lang(\spec)$. A specification $\spec$ is
\emph{realizable} if there exists a design $\design$ that realizes it.
%
%
%
A \emph{safety} specification $\spec^s$ is represented by an automaton
$\spec^s = (\states, \init, \dalph, \delta, F)$, where $\dalph =
\dinalph\cup\doutalph$, $\delta : \states \times \dalph \rightarrow
\states$, and $F\subseteq \states$ is a set of safe states.  The
\emph{run} induced by trace $\dtrace = \dletter_0 \dletter_1 \ldots
\in \dalph^\infty$ is the state sequence $\overline{q} = q_0 q_1
\ldots $ such that $q_{i+1} = \delta(q_i, \dletter_i)$; the run is
\emph{accepting} if $\forall i\geq 0 \scope q_i \in F$.  Trace
$\dtrace$ (of a design $\design$) \emph{satisfies} $\spec^s$ if the
induced run is accepting. The \emph{language}
$\lang(\spec^s)$ is the set of all traces satisfying $\spec^s$.

\noindent
\textbf{Games.}
A (2-player, alternating) \emph{game} is a tuple $\game = (\gstates,
\ginit, \dinalph, \doutalph, \delta, \win)$,
where $\gstates$ is a finite set of game states, $\ginit \in \gstates$ is the initial state,
$\delta: \gstates \times \dinalph \times \doutalph \rightarrow \gstates$
is a complete transition function, and $\win: \gstates^\omega
\rightarrow \B$ is a winning condition.  The game is played by two
players: the system and the environment.  In every state $g\in \gstates$
(starting with $\ginit$), the environment first chooses an input letter
$\dinletter \in \dinalph$, and then the system chooses some output
letter $\doutletter \in \doutalph$. This defines the next state $g' =
\delta(g,\dinletter, \doutletter)$, and so on. Thus, a (finite or
infinite) word over $\Sigma$ results in a (finite or infinite) \emph{play}, a
sequence $\overline{g} = g_0 g_1 \ldots$ of game states.  A play is \emph{won} by the system iff
$\win(\overline{g})$ is $\true$.
A \emph{safety game} defines $\win$ via a set $F^s\subseteq \gstates$ of
safe states: $\win(g_0 g_1 \ldots)$ is $\true$ iff $\forall i \geq 0
\scope g_i \in F^s$, i.e., if only safe states are visited.
Let $\inf(\overline{g})$ denote the states that occur infinitely often in $\overline{g}$. A \emph{\buchi game} defines $\win$ via a set $F^b\subseteq \gstates$ of accepting states: $\win(\overline{g})$ is $\true$ iff $\inf(\overline{g}) \cap F^b \neq \emptyset$.

It is easy to transform a safety specification into a safety game such
that a trace satisfies the specification iff the corresponding play is
won. Given a safety specification $\spec^s$. A finite trace $\dtrace
\in \dalph^*$ is \emph{wrong}, if the corresponding play is not won,
i.e., if there is no way for the system to guarantee that any
extension of the trace satisfies the specification. An \emph{output}
is called \emph{wrong}, if it makes a trace wrong; i.e., given
$\spec^s$, a trace $\dtrace \in \dalph^*$ an input $\dinletter \in
\dinalph$, and an output $\doutletter \in \doutalph$, $\doutletter$ is
wrong iff $\dtrace$ is not wrong, but $\dtrace \cdot
(\dinletter,\doutletter)$ is.

A deterministic (memoryless) \emph{strategy} for the environment is a function
$\rho_e: \gstates \rightarrow \dinalph$.
A deterministic (memoryless) \emph{strategy} for the system is a function $\rho_s:
\gstates \times \dinalph \rightarrow \doutalph$. A strategy $\rho_s$ is
\emph{winning} for the system, if \emph{for all} strategies $\rho_e$ of the environment
the play $\overline{g}$ that is constructed when defining the outputs using
$\rho_e$ and $\rho_s$ satisfies $\win(\overline{g})$.
The \emph{winning region} $W$ is the set of states
from which a winning strategy exists.
%
A strategy is \emph{cooperatively winning} if there \emph{exists} a strategy $\rho_e$
and $\rho_s$, such that the play $\overline{g}$
constructed by $\rho_e$ and $\rho_s$ satisfies $\win(\overline{g})$.

For a \buchi game $\mathcal{G}$ with accepting states $F^b$,
consider a strategy $\rho_e$ of the environment, a strategy $\rho_s$ of the system, and a state $g\in G$.
We set the distance $dist(g, \rho_e, \rho_s)=k$, if the play $\overline{g}$
defined by $\rho_e$ and $\rho_s$ reaches from $g$ an accepting state
that occurs infinitely often in $\overline{g}$ in $k$ steps.
If no such state is visited, we set $dist(g, \rho_e, \rho_s)=\infty$.
Given two strategies $\rho_s$ and  $\rho_s'$ of the system, we say that $\rho_s'$
\emph{dominates} $\rho_s$ if: (i) for all $\rho_e$ and all $g\in G$,
$dist(g,\rho_e,\rho_s')\leq dist(g,\rho_e,\rho_s)$ , and (ii) there exists
$\rho_e$ and $g\in G$ such that $dist(g,\rho_e,\rho_s')< dist(g,\rho_e,\rho_s)$.

A strategy is \emph{admissible} if there is no strategy that dominates it.


\section{Admissible Shields}
\label{sec:def}

Bloem et al.~\cite{BloemKKW15} presented the general framework for shield synthesis.
A shield has two main properties: (i) For any design, a shield ensures \emph{correctness} with respect to a specification. (ii) A shield ensures \emph{minimal deviation}. We revisit these properties
in Sec. \ref{sec:def_shields}.
The definition of minimum deviation is designed to be flexible
and different notions of minimum deviation can be realized.
$k$-stabilizing shields represent one concrete realization.
In Sec. \ref{sec:def_admissible_shields}, we present a new
realization of the minimum deviation property resulting in admissible shields.

\subsection{Definition of Shields}
\label{sec:def_shields}

A shield reads the input and output of a design as shown in Fig.~\ref{fig:attach_shield}.
We then address the two properties, correctness and minimum deviation, to be ensured by a shield.

\textbf{The Correctness Property.}  With correctness we refer to the
property that the shield corrects any design's output such that a
given safety specification is satisfied. Formally, let $\spec$ be a
safety specification and $\shield = (\states', \init', \dalph,
\doutalph, \delta', \lambda')$ be a Mealy machine. We say that
$\shield$ \emph{ensures correctness} if for any design $\design =
(\states, \init, \dinalph, \doutalph, \delta, \lambda)$, it holds that
$(\design \comp \shield) \models \spec$.

Since a shield must work for any design, the synthesis procedure
does not need to consider the design's implementation. This property
is crucial because the design may be unknown or too complex to
analyze. On the other hand, the design may satisfy additional
(noncritical) specifications that are not specified in $\spec$ but
should be retained as much as possible.

\textbf{The Minimum Deviation Property.}  Minimum deviation requires a
shield to deviate only if necessary, and as infrequently as possible.
To ensure minimum deviation, a shield can only deviate from the design
if a property violation becomes unavoidable.  Given a safety
specification $\spec$, a Mealy machine $\shield$ \emph{does not
  deviate unnecessarily} if for any design $\design$ and any trace
$\dintrace||\douttrace$ that is not wrong, we have that
$\shield(\dintrace||\douttrace) = \douttrace$. In other words, if
$\design$ does not violate $\spec$, $\shield$ keeps the output of
$\design$ intact.

A Mealy machine $\shield$ is a \emph{shield} if $\shield$ ensures
correctness and does not deviate unnecessarily.

Ideally, shields end phases of deviations as soon as possible,
recovering quickly.  This property leaves room for
interpretation. Different types of shields differentiate on how this
property is realized.

\subsection{Defining Admissible Shields}
\label{sec:def_admissible_shields}

In this section we define admissible shields using their speed of
recovery.  We distinguish between two situations.  In states of the
design in which a finite number $k$ of deviations can be guaranteed,
an admissible shield takes an adversarial view on the design:
it guarantees recovery within $k$ steps regardless of system behavior,
for the smallest $k$ possible. In
these states, the strategy of an admissible shield conforms to the
strategy of $k$-stabilizing shield.  In all other states, admissible
shields take a collaborative view: the admissible shield will attempt
to work with the design to recover as soon as possible.  In
particular, an admissible shield plays an admissible strategy, that
is, a strategy that cannot be beaten in recovery speed if the design
acts cooperatively.

We will now define admissible shields.  For failures of the system
that are corrected by the shield, we consider four phases:
\begin{enumerate}
\item The \emph{innocent phase} consisting of inputs $\dintrace$ and outputs
$\douttrace$, in which no failure occurs; i.e., $(\dintrace||\douttrace) \models \spec$.
\item The \emph{misstep phase} consisting of a input $\dinletter$ and a wrong
output $\doutletter^f$; i.e., $(\dintrace||\douttrace) \cdot (\dinletter, \doutletter^f) \not\models \spec$.
\item The \emph{deviation phase} consisting of inputs $\dintrace'$ and outputs  $\douttrace'$ in which the shield is allowed to deviate, and
for a correct output $\doutletter^c$ we have $(\dintrace||\douttrace) \cdot (\dinletter,\doutletter^c)\cdot (\dintrace'||\douttrace') \models \spec$.
\item The \emph{final phase} consisting $\dintrace''$ and $\douttrace''$ in which the shield does not
deviate, and $(\dintrace||\douttrace) \cdot(\dinletter,\doutletter^c)\cdot (\dintrace'||\douttrace') \cdot
(\dintrace''||\douttrace'') \models \spec$.
\end{enumerate}
Adversely $k$-stabilizing shields have a deviation phase of length at most $k$.

\begin{definition}
 A shield $\shield$ adversely $k$-stabilizes a trace $\dtrace =
 \dintrace || \douttrace \in \dalph^*$, if
 for any input $\dinletter\in\dinalph$ and any wrong output $\doutletter^f\in\doutalph$,
  \emph{\textbf{for any}} correct output $\doutletter^c\in\doutalph$
   and \emph{\textbf{for any}} correct trace $\dintrace' || \douttrace' \in \dalph^k$
   there exists a trace $\doutletter^\# \douttrace^\# \in\doutalph^{k+1}$ such that
   for any trace $\dintrace'' || \douttrace'' \in \dalph^{\omega}$
  such that $(\dintrace||\douttrace) \cdot(\dinletter,\doutletter^c)\cdot (\dintrace'||\douttrace') \cdot
(\dintrace''||\douttrace'') \models \spec$,
we have
$$ \shield(\dtrace \cdot (\dinletter,\doutletter^f)\cdot  (\dintrace' || \douttrace') \cdot
(\dintrace'' || \douttrace'')) = \douttrace \cdot \doutletter^\# \cdot \douttrace^\# \cdot
\douttrace''
$$
and
$$(\dintrace||\douttrace) \cdot(\dinletter,\doutletter^\#) \cdot (\dintrace'||\douttrace^\#) \cdot
(\dintrace''||\douttrace'') \models \spec.$$
\end{definition}

Note that it is not always possible to adversely $k$-stabilize a
shield for a given $k$ or even for any $k$.

\begin{definition} [Adversely $k$-Stabilizing Shields
\cite{BloemKKW15}] A shield $\shield$ is adversely $k$-stabilizing if
  it adversely $k$-stabilies any finite trace.
\end{definition}

An adversely $k$-stabilizing shield guarantees to end deviations after
at most $k$ steps and produces a correct trace under the assumption
that the failure of the design consists of a transmission error in the
sense that the wrong letter is substituted for a correct one. We use
the term \emph{adversely} to emphasize that finitely long deviations
are guaranteed for \emph{any} future inputs and outputs of the design.

\begin{definition}[Adversely Subgame Optimal Shield]
A shield $\shield$ is \emph{adversely subgame optimal} if for any
trace $\dtrace \in \dalph^*$, $\shield$ adversely $k-$stabilizes
$\dtrace$ and there exists no shield that adversely $l$-stabilizes
$\dtrace$ for any $l<k$.
\end{definition}
An adversely subgame optimal shield $\shield$ guarantees to deviate in
response to an error for at most $k$ time steps, for the smallest $k$
possible.

\begin{definition}
 A shield $\shield$ collaboratively $k$-stabilizes a trace $\dtrace =
 \dintrace || \douttrace \in \dalph^*$, if
 for any input $\dinletter\in\dinalph$ and any wrong output $\doutletter^f\in\doutalph$,
  \emph{\textbf{there exists}} a correct output $\doutletter^c\in\doutalph$,
  a correct trace $\dintrace' || \douttrace' \in \dalph^k$,
   and a trace $\doutletter^\# \douttrace^\# \in\doutalph^{k+1}$ such that
   for any trace $\dintrace'' || \douttrace'' \in \dalph^{\omega}$
  such that $(\dintrace||\douttrace) \cdot(\dinletter,\doutletter^c)\cdot (\dintrace'||\douttrace') \cdot
(\dintrace''||\douttrace'') \models \spec$,
we have
$$ \shield(\dtrace \cdot (\dinletter,\doutletter^f)\cdot  (\dintrace' || \douttrace') \cdot
(\dintrace'' || \douttrace'')) = \douttrace \cdot \doutletter^\# \cdot \douttrace^\# \cdot
\douttrace''
$$
and
$$(\dintrace||\douttrace) \cdot(\dinletter,\doutletter^\#) \cdot (\dintrace'||\douttrace^\#) \cdot
(\dintrace''||\douttrace'') \models \spec.$$
\end{definition}

\begin{definition} [Collaborative $k$-Stabilizing Shield] A shield $\shield$ is
  collaboratively
  $k$-stabilizing if it collaboratively $k$-stabilizes
  any finite trace.
\end{definition}

A collaborative $k$-stabilizing shield requires that it must be
possible to end deviations after $k$ steps, for some future input and
output of $\design$.  It is not necessary that this is possible for
all future behavior of $\design$ allowing infinitely long deviations.

\begin{definition}[Collaborative Subgame Optimal Shield]
A shield $\shield$ is \emph{collaborative subgame optimal} if for any
trace $\dtrace \in \dalph^*$, $\shield$ collaboratively $k-$stabilizes
$\dtrace$ and there exists no shield that adversely $l$-stabilizes
$\dtrace$ for any $l<k$.
\end{definition}

\begin{definition}[Admissible Shield]
A shield $\shield$ is admissible if for any trace $\dtrace$,
whenever there exists a $k$ and a shield $\shield'$ such that
$\shield'$ adversely $k$-stabilizes $\dtrace$, then $\shield$
adversely $k$-stabilizes $\dtrace$. If such a $k$ does not exist for
trace $\dtrace$, then $\shield$  collaboratively $k$-stabilizes
$\dtrace$ for a minimal $k$.
\end{definition}

An admissible shield ends deviations whenever possible.
In all states of the design $\design$ where a finite number of deviations can be guaranteed,
an admissible shield deviates for each violation for at most $k$
steps, for the smallest $k$ possible. In all other states, the shield
corrects the output in such a way that there exists design's inputs
and outputs such that deviations end after $l$ steps,
for the smallest $l$ possible.

\section{Synthesizing Admissible Shields}
\label{sec:sol}

The flow of the synthesis procedure is illustrated in
Fig.~\ref{fig:c-stab}.
Starting from a safety specification $\spec = (Q, q_{0}, \dalph, \delta,F)$
with $\dalph=\dinalph\times\doutalph$, the admissible shield synthesis procedure
consists of five steps.

\begin{figure}[tb]
  \begin{center}
    \includegraphics[width=0.85\textwidth]{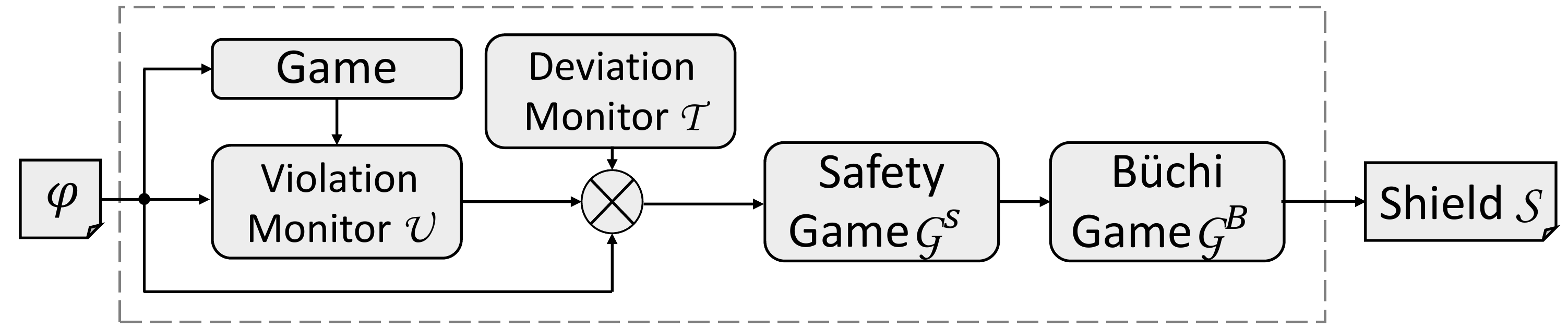}
    \caption{Outline of our admissible shield synthesis procedure.}
    \label{fig:c-stab}
  \end{center}
\end{figure}

\noindent
\subsubsection{Step 1. Constructing the Violation Monitor $\mathcal{U}$.}
\label{sec:violation_monitor}
From $\spec$ we build the automaton $\mathcal{U} = (U, u_0, \dalph, \delta^u)$ to monitor property
violations by the design.  The goal is to identify the latest point in
time from which a specification violation can still be corrected with a
deviation by the shield. This constitutes the start of the \emph{recovery}
period, in which the shield is allowed to deviate from the design.
In this phase the shield monitors the design from all states that the design could reach
under the current input and a correct output.
A second violation occurs only if the next design's output is inconsistent with all
states that are currently monitored. In case of a second violation,
the shield monitors the set of all input-enabled states that are reachable from the current set of monitored states.

The first phase of the construction of the violation monitor $\mathcal{U}$ considers
$\spec = (Q, q_{0}, \\ \dalph, \delta,F)$  as a
\emph{safety game} and computes its winning region $W\subseteq F$ so that every reactive system
$\design\models\spec$ must produce outputs such that the next state of $\spec$ stays in $W$.  Only
in cases in which the next state of $\spec$ is outside of $W$ the shield is allowed to interfere.

The second phase expands the state space $Q$ to
$2^{Q}$ via a subset construction, with the following rationale.
If the design makes a mistake (i.e., picks outputs such that
$\spec$ enters a state $q\not \in W$), we have to ``guess'' what the design
actually meant to do and we consider all output letters that would
have avoided leaving $W$ and continue monitoring the design
from all the corresponding successor states in parallel.  Thus,
$\mathcal{U}$ is essentially a subset construction of $\spec$,
where a state $u\in U$ of $\mathcal{U}$ represents a set of states in
$\spec$.

The third phase expands the state space of $\mathcal{U}$
by adding a counter $d\in\{0,1,2\}$ and a
output variable $z$.
Initially $d$ is 0.
Whenever a property is violated
$d$ is set to 2.
If $d>0$, the shield is in the recovery phase and can deviate.
If $d=1$ and there is no other violation, $d$ is decremented to 0.
In order to decide when to decrement $d$ from 2 to 1,
we add an output $z$ to the shield. If this output is set to
$\true$ and $d = 2$, then $d$ is set to 1.

The final violation monitor is $\mathcal{U} = (U, u_0, \dalph^u,
\delta^u)$, with the set of states $U = (2^{Q} \times \{0,1,2\})$,
the initial state $u_0 = (\{q_0\}, 0)$,
the input/output alphabet $\dalph^u = \dinalph \times \doutalph^u$
with $\doutalph^u = \doutalph \cup z$, and the next-state function $\delta^u$ , which obeys the following rules:
\begin{enumerate}
\item $\delta^u((u,d), (\dinletter, \doutletter)) =
      \bigl(\{q' \kin W \mid \exists q\in u, \doutletter' \in \doutalph^u
         \scope \delta(q,(\dinletter,\doutletter')) = q'\}, 2\bigr)$\\
if $\forall q \in u \scope
    \delta(q,(\dinletter, \doutletter)) \not\in W$,
and
\label{eq:subset_m}
\item $\delta^u((u,d), \dletter) \!=\!
     \bigl(\{q'\kin W \mid \exists q\kin u \scope\delta(q,\dletter) = q'\},
       \textsf{dec}(d)\bigr)$
if $\exists q \kin u \scope \delta(q,\dletter) \kin W$,
and $\textsf{dec}(0) = \textsf{dec}(1) = 0$, and if
$z$ is $\true$ then $\textsf{dec}(2) = 1$,
else $\textsf{dec}(2) = 2$.
\label{eq:subset_n}
\end{enumerate}
Our construction sets $d=2$ whenever the design leaves the winning
region, and not when it enters an unsafe state.  Hence, the shield
$\shield$ can take a remedial action as soon as ``the crime is
committed'', before the damage is detected, which would have been too
late to correct the erroneous outputs of the design.

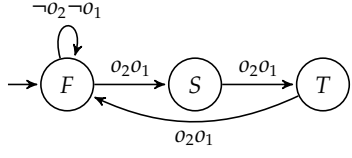
\begin{figure}[tb]
\vspace{-18pt}
\begin{minipage}{\linewidth}
  \centering
  \begin{minipage}{0.46\linewidth}
    \begin{figure}[H]
        \centering
          \centering\begin{tikzpicture}
  \node[state,draw,inner sep=1pt,minimum size=0.7cm, initial] (s0) at (0,0) {$F$};
  \node[state,draw,inner sep=1pt,minimum size=0.7cm] (s1) at (1.7,0) {$S$};
  \node[state,draw,inner sep=1pt,minimum size=0.7cm] (s2) at (3.4,0) {$T$};
  \draw (s0) edge[->] node[above] {$o_2o_1$} (s1);
  \draw (s1) edge[->] node[above] {$o_2o_1$} (s2);
  \draw (s2) edge[->, bend left=25] node[below] {$o_2o_1$} (s0);

  \draw (s0) edge[loop above] node[above] {$\neg o_2 \neg o_1$} (s0);
  \end{tikzpicture}
        \caption{Safety automaton of Example \ref{ex:monitor_U}.}
        \label{fig:ex_spec}
    \end{figure}
  \end{minipage}
  ~
  \begin{minipage}{0.49\linewidth}
    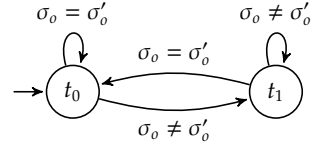
\begin{figure}[H]
        \centering
        \centering\begin{tikzpicture}
  \node[state,draw,inner sep=1pt,minimum size=0.7cm, initial] (s0) at (0,0) {$t_0$};
  \node[state,draw,inner sep=1pt,minimum size=0.7cm] (s1) at (2.7,0) {$t_1$};
  \draw (s0) edge[->, bend right=15] node[below] {$\sigma_o\neq \sigma_o'$} (s1);
  \draw (s1) edge[->, bend right=15] node[above] {$\sigma_o=\sigma_o'$} (s0);

  \draw (s0) edge[loop above] node[above] {$\sigma_o=\sigma_o'$} (s0);
  \draw (s1) edge[loop above] node[above] {$\sigma_o\neq\sigma_o'$} (s1);
  \end{tikzpicture}
        \caption{The deviation monitor $\mathcal{T}$.}
        \label{fig:dev_monitor}
    \end{figure}
  \end{minipage}
\end{minipage}
\end{figure}

\begin{exa}
\label{ex:monitor_U}
We illustrate the construction of $\mathcal{U}$ using the specification
$\spec$ from Fig.~\ref{fig:ex_spec} over the outputs $o_1$ and $o_2$.
(Fig.~\ref{fig:ex_spec} represents a safety automaton if we make all missing edges point to an
(additional) unsafe state.)
 The winning region consists of all safe
states, i.e., $W = \{F,S,T\}$.  The resulting violation monitor is
$\mathcal{U}=
(\{\text{F},\text{S},\text{T},\text{FS},\text{ST},\text{FT},\\\text{FST}\}
\times\{0,1,2\}, (\text{F},0), \dalph^u, \delta^u)$.
The transition relation $\delta^u$
is illustrated in Table~\ref{fig:ex1_table}
and lists the next states for all possible
present states and outputs.
Lightning bolts denote specification
violations.
The update of counter $d$, which is not included in
Table~\ref{fig:ex1_table}, is as follows: Whenever the design
commits a violation 
$d$ is set to $2$. If no violation exists, $d$ is decremented in the following way:
if $d=1$ or $d=0$, $d$ is set to 0. If $d=2$ and $z$ is $\true$, $d$ is set to 1, else $d$ remains 2. In this example, $z$ is set to $\true$,
whenever we are positive about the current state of the design (i.e. in $(\{F\},d)$,
$(\{S\},d)$, and $(\{T\},d)$).

Let us take a closer look at some entries of Table~\ref{fig:ex1_table}.
If the current state is $(\{F\},0)$ and we observe
the output $\neg o_2 o_1$, a specification violation occurs.
We assume that $\design$  meant to give an allowed output, either
$o_2 o_1$ or $\neg o_2 \neg o_1$. The shield continues
to monitor both $F$ and $S$; thus, $\mathcal{U}$ enters the state $(\{F,S\},2)$.
If the next observation is $o_2 o_1$, which is allowed from both
possible current states, the possible next states are $S$ and $T$, therefore
$\mathcal{U}$ traverses to state $(\{S,T\},2)$.
However, if the next observation is again $\neg o_2 o_1$,
which is neither allowed in $F$ nor in $S$, we know that a second violation occurs.
Therefore, the shield monitors the design from all three states
and $\mathcal{U}$ enters the state $(\{F,S,T\},2)$.

\begin{table}[tb]
\caption{$\delta^u$ of $\mathcal{U}$ of Example \ref{ex:monitor_U}.}
\label{fig:ex1_table}
\centering
\begin{tabular}{l|ccc}
    &$\neg o_1 \neg o_2$   &$\neg o_1 o_2$ or $o_1\neg o_2$  &$o_1o_2$ \\
\hline
\{F\}   &\{F\}      &\{F,S\}\err   &\{S\}   \\
\{S\}   &\{T\}\err  &\{T\}\err    &\{T\}   \\
\{T\}   &\{F\}\err  &\{F\}\err    &\{F\}   \\
\{F,S\}  &\{F\}      &\{F,S,T\}\err   &\{S,T\}   \\
\{S,T\}  &\{F,T\}\err &\{F,T\}\err   &\{F,T\}   \\
\{F,T\}  &\{F\}      &\{F,S,T\}\err  &\{F,S\}   \\
\{F,S,T\} &\{F\}      &\{F,S,T\}\err  &\{F,S,T\}   \\
\hline
\end{tabular}
\end{table}

\end{exa}

\noindent
\subsubsection{Step 2. Constructing the Deviation Monitor $\mathcal{T}$.}
\label{sec:deviation_monitor}
We build $\mathcal{T} = (T, t_0, \doutalph \times \doutalph, \delta^t)$
to monitor deviations between the shield and design outputs.
Here, $T = \{t_0, t_1\}$ and $\delta^t(t, (\doutletter,
\doutletter')) = t_0$ iff $\doutletter = \doutletter'$. That is, if there is a deviation in the current time step, then $\mathcal{T}$ will be in $t_1$ in the next time step. Otherwise, it will be in $t_0$.  This deviation monitor is shown in
Fig.~\ref{fig:dev_monitor}.

\noindent
\subsubsection{Step 3. Constructing and Solving the Safety Game $\mathcal{G}^s$.}
Given the automata $\mathcal{U}$ and $\mathcal{T}$ and the
safety automaton $\spec$, we construct a safety
game $\mathcal{G}^s = (G^s, g_0^s, \dinalph^s, \doutalph^s$
$\delta^s, F^s)$, which is the synchronous product of $\mathcal{U}$,
$\mathcal{T}$, and $\spec$, such that $G^s= U \times
T \times Q$ is the state space, $g_0^s = (u_0, t_0, q_0)$
is the initial state, $\dinalph^s=\dinalph\times\doutalph$ is the input of the
shield, $\doutalph^s=\doutalph\cup \{z\}$ is the output of the shield,
$\delta^s$ is the next-state function, and $F^s$ is the set of safe states such that
$\delta^s\bigl((u, t, q), (\dinletter, \doutletter),
(\doutletter',z)\bigr) = $
\[
\bigl( \delta^u[u,(\dinletter, (\doutletter,z))],
       \delta^t[t,(\doutletter, \doutletter')],
       \delta[q, (\dinletter, \doutletter')]
\bigr),
\]

\noindent and $F^s = \{(u, t, q)\in G^s \mid
             q \in F \wedge u=(w,0) \rightarrow t=t_0\}$.

We require $q \in F$, which ensures that the
output of the shield satisfies $\spec$, and that the shield can
only deviate in the recovery period (i.e., if $d=0$, no deviation is allowed).
We use standard algorithms for safety games (cf.
\cite{Faella09}) to compute the winning region $W^s$ and the most permissive non-deterministic winning strategy $\rho_s: \gstates \times \dinalph \rightarrow 2^{\doutalph}$ that is not only winning for the system, but also contains all deterministic winning strategies.

\noindent
\subsubsection{Step 4. Constructing the \buchi Game $\mathcal{G}^b$.}
Implementing the safety game ensures correctness ($\design \comp \shield \models \spec$)
and that the shield $\shield$ keeps the output of the design $\design$ intact,
if $\design$ does not violate $\spec$.
The shield still has to keep
the number of deviations per violation to a minimum.
Therefore, we would like the recovery period to be over infinitely often.
This can be formalized as a \buchi winning condition.
We construct the \buchi game $\mathcal{G}^b$
by applying the non-deterministic safety strategy $\rho^s$
to the game graph $\mathcal{G}^s$.

Given the safety game $\mathcal{G}^s=(G^s, g_0^s, \dinalph^s, \doutalph^s, \delta^s, F^s)$ with the
non-deterministic winning strategy $\rho^s$ and the winning region $W^s$, we construct a \buchi game
$\mathcal{G}^b=(G^b, g_0^b, \dinalph^s, \\\doutalph^s, \delta^b, F^b)$
such that $G^b=W^s$ is the state space, the initial state $g_0^b=g_0^s$ and
the input/output alphabet $\dinalph^b=\dinalph^s$ and $\doutalph^b=\doutalph^s$
remain unchanged, $\delta^b=\delta^s\cap\rho^s$ is the transition
function, and $F^b = \{(u, t, q)\in W^s \mid
             (u=(w,0) \vee u=(w,1))\}$ is the set of accepting states.
A play is winning if $d\leq1$ infinitely often.

\noindent
\subsubsection{Step 5. Solving the \buchi Game $\mathcal{G}^b$.}
\label{sec:solving_buchi}
Most likely, the \buchi game $\mathcal{G}^b$
contains reachable states, for which $d\leq1$ cannot be enforced infinitely often.
We implement an admissible strategy
that enforces to visit $d\leq1$ infinitely often whenever possible.
This criterion essentially asks for a strategy
that is winning with the help of the design.

The admissible strategy $\rho^b$ for a \buchi game $\mathcal{G}^b=(G^b, g_0^b, \dinalph^b, \doutalph^b, \delta^b, F^b)$ can be computed as follows\cite{Faella09}:
\begin{enumerate}
  \item Compute the winning region $W^b$ and a winning strategy $\rho_w^b$ for $\mathcal{G}^b$ (cf.\ \cite{Mazala01}).

  \item Remove all transitions that start in $W^b$ and do not belong to $\rho_w^b$ from  $\mathcal{G}^b$.
      This results in a new \buchi game
      $\mathcal{G}_1^b=(G^b, g_0^b, \dinalph^b, \doutalph^b, \delta_1^b, F^b)$
      with $(g,(\dinletter,\doutletter),g')\in\delta_1^b$
      if $(g,\dinletter, \doutletter)\in\rho_w^b$ or if
      $\forall \doutletter' \in \doutalph^b \scope (g,\dinletter, \doutletter')\notin\rho_w^b \wedge (g,(\dinletter,\doutletter),g')\in\delta^b$.

  \item In the resulting game $\mathcal{G}_1^b$, compute a cooperatively winning strategy $\rho^b$.
        In order to compute $\rho^b$, one first has to transform all input variables to output variables. This results in the \buchi game
        $\mathcal{G}_2^b=(G^b, g_0^b, \emptyset, \dinalph^b\times\doutalph^b, \delta_1^b, F^b)$. Afterwards, $\rho^b$ can be computed with the standard algorithm for the winning strategy
        on $\mathcal{G}_2^b$.
\end{enumerate}

The strategy $\rho^b$ is an admissible strategy of the game $\mathcal{G}^b$,
since it is winning and cooperatively winning~\cite{Faella09}.
Whenever the game $\mathcal{G}^b$ starts in a state of the winning region $W^b$,
any play created by $\rho_w^b$ is winning.
Since $\rho^b$ coincides with $\rho_w^b$ in all states of the winning region $W^b$,
$\rho^b$ is winning.
We know that $\rho^b$ is cooperatively winning in the game $\mathcal{G}_1^b$.
A proof that $\rho^b$ is also cooperatively winning in the original game $\mathcal{G}^b$
can be found in \cite{Faella09}.

\begin{theorem}
A shield that implements the admissible strategy $\rho^b$
in the \buchi game $\mathcal{G}^b=(G^b, g_0^b, \dinalph^b, \doutalph^b, \delta^b, F^b)$ in a new reactive system
$\shield = (G^b, g^b_0, \dinalph^b, \doutalph^b,
\delta', \rho^b)$ with $\delta'(g,\dinletter) = \empty
\delta^b(g,\dinletter,\rho^b(g,\dinletter))$ is an admissible shield.
\end{theorem}
\begin{proof}
  First, the admissible strategy $\rho^b$ is winning for all winning states of the \buchi game $\mathcal{G}^b$. Since winning strategies for \buchi games are subgame optimal,
  a shield that implements $\rho^b$ ends deviations after the smallest number of steps possible,
  for all states of the design in which a finite number of deviations can be guaranteed.
 Second, $\rho^b$ is cooperatively winning in the \buchi game $\mathcal{G}^b$.
  Therefore, in all states in which a finite number of deviation cannot be guaranteed, a shield that implements the strategy $\rho^b$ recovers with the help of the design as soon as possible.
\end{proof}

The standard algorithm for solving \buchi games
contains the computation of attractors; the $i$-th attractor for the system contains all states  from which the system can ``force'' a visit of an accepting
state in $i$ steps. For all states $g\in G^b$ of the game $\mathcal{G}^b$, the attractor
number of $g$ corresponds to the smallest number of steps within
which the recovery phase can be guaranteed to end, or can end with the help of the design
if a finite number of deviation cannot be guaranteed.

\begin{theorem}
  Let $\spec=\{Q, q_{0}, \dalph, \delta, F\}$ be a safety specification and
  $|Q|$ be the cardinality of the state space of $\spec$.
  An admissible shield with respect to $\spec$ can be synthesized in $\mathcal{O}(2^{|Q|})$ time, if it exists.
\end{theorem}
\begin{proof}
  Our safety game $\mathcal{G}^s$  and our \buchi game $\mathcal{G}^b$ have at most
  $m=(2 \cdot 2^{|Q|}+|Q|)\cdot 2  \cdot |Q|$ states and at most $n=m^2$ edges.
  Safety games can be solved in $\mathcal{O}(m+n)$ time
  and \buchi games in $\mathcal{O}(m\cdot n)$ time~\cite{Mazala01}.
\end{proof}

\section{Experimental Results}
\label{sec:exp}

We implemented our admissible shield synthesis procedure 
in Python, which takes a set of safety automata
defined in a textual representation as input.
The first step in our synthesis procedure is to build the product
of all safety automata and construct the violation monitor \ref{sec:violation_monitor}.
This step is performed on an explicit representation.
For the remaining steps
we use Binary Decision Diagrams (BDDs)  for symbolic representation.
The synthesized shields are encoded in Verilog format.
%
%
To evaluate the performance of our tool, we constructed three sets of experiments,
the basis of which is the safety specification of Fig. \ref{fig:map}. This example represents a map with 15 waypoints and the six safety properties \ref{connected}-\ref{home}.
First, we reduced the complexity of the
example by only considering 8 out of 15 waypoints.
This new example, called \emph{Map$_8$}, consists of the
waypoints $loc_1$ to $loc_8$ with their corresponding properties.
The second series of experiments, called \emph{Map$_{15}$},
considers the original specification of Fig. \ref{fig:map} over
all 15 waypoints. The synthesized shields behave as
described in Section~\ref{sec:ex}.
The third series of experiments, called \emph{Map$_{31}$},
considers a map with $31$ waypoints, essentially adding
a duplicate of the map in Fig. \ref{fig:map}.
All results are summarized in Table~\ref{tab:res1} and in Table~\ref{tab:res2}.
For both tables, the first columns list the set of specification automata and the number
of states, inputs, and outputs of their product automata.
The next column lists the smallest number of steps $l$
under which the shield is able to recover with the help
of the design. The last column lists the
synthesis time in seconds.
All computation times are for a
computer with a 2.6 GHz Intel i5-3320M CPU with 8 GB RAM
running an 64-bit distribution of Linux.
Source code, input files,
and 
instructions to reproduce our experiments are available for
download\footnote{\scriptsize\url{
http://www.iaik.tugraz.at/content/research/design_verification/others/}
}.

\begin{figure}[tb]
\vspace{-18pt}
\begin{minipage}{\linewidth}
  \centering
  \begin{minipage}{0.46\linewidth}
\begin{table}[H]
\caption{Results of $map_8$ and $map_{31}$.}
\label{tab:res1}
\centering
\begin{tabular}{c|l|ccccc}

\centering
 Example         &Property          &$|Q|$ &$|I|$    &$|O|$    &$l$ &Time [sec] \\
\hline
Map$_8$     &1              &9    &0        &3         &3   &0.52        \\
            &1+4            &12   &0        &3         &3   &1.2        \\
            &1+5a           &46   &1        &3         &4   &6.2        \\
            &1+5b           &32   &1        &3         &3   &7        \\
            &1+4+5a         &55   &1        &3         &4   &17        \\
            &1+4+5b         &36   &1       &3          &3   &12        \\
\hline
Map$_{31}$  &1              &32   &0        &5         &6   &122        \\
Map$_{31}$  &1+2            &32   &0        &5         &6   &143        \\
Map$_{31}$  &1+2+3          &34   &0        &5         &6   &183        \\
Map$_{31}$  &1+2+3+4        &38   &0        &5         &6   &238        \\

\hline
\end{tabular}
\end{table}
  \end{minipage}
  ~~
  \begin{minipage}{0.46\linewidth}
\begin{table}[H]
\caption{Results of $map_{15}$.}
\label{tab:res2}
\centering
\begin{tabular}{c|l|ccccc}

\centering
  Example         &Property          &$|Q|$ &$|I|$    &$|O|$    &$l$ &Time [sec] \\
\hline
Map$_{15}$  &1              &16   &0        &4         &5   &12        \\
            &1+2            &16   &0        &4         &5   &14        \\
            &1+2+3          &19   &0        &4         &5   &19        \\
            &1+2+3+4        &23   &0        &4         &5   &28        \\
            &1+5a           &84   &1        &4         &6   &173          \\
            &1+5a+2         &84   &1        &4         &6   &205          \\
            &1+5a+2+3       &100  &1        &4         &6   &307          \\
            &1+5b           &64   &1        &4         &6   &169          \\
            &1+5b+2         &64   &1        &4         &6   &195          \\
            &1+6            &115  &1        &4         &7   &690          \\

\hline
\end{tabular}
\end{table}
  \end{minipage}
\end{minipage}
\end{figure}

\section{Conclusion}
\label{sec:conc}

We have proposed a new shield synthesis procedure to
synthesize \emph{admissible shields}.
We have shown that admissible shields overcome the limitations of previously developed $k$-stabilizing shields.
We believe our approach and first experimental results
over our case study involving UAV mission planning open several
directions for future research. At the moment, shields only attend to
safety properties and disregard liveness properties. Integrating liveness
is therefore a preferable next step.
Furthermore, we plan to further develop our prototype tool and
apply shields in other domains such as in the
distributed settings or for Safe Reinforcement Learning, in which safety constraints
must be enforced during the learning processes. We plan to investigate
how a shield might be most beneficial in such settings.




\bibliography{main}

\end{document}